\documentclass[draftclsnofoot,onecolumn,12pt,journal]{IEEEtran}

\usepackage{times}
\usepackage{epsfig}
\usepackage{graphicx}
\usepackage{amsmath}
\usepackage{amssymb}
\usepackage{amsthm}
\usepackage{enumerate}
\usepackage{indentfirst}
\usepackage{subfigure}
\usepackage{algorithmic}
\usepackage{algorithm}
\usepackage{url}
\usepackage{booktabs}
\usepackage{multirow}
\usepackage{longtable}
\usepackage{cite}

\theoremstyle{remark}
\newtheorem{theorem}{Theorem}

\newtheorem{lemma}{Lemma}
\newtheorem{corollary}{Corollary}

\long\def\symbolfootnote[#1]#2{\begingroup\def\thefootnote{\fnsymbol{footnote}}
\footnote[#1]{#2}\endgroup}

\ifCLASSINFOpdf
\else
\fi
\begin{document}
%
\title{Random Beamforming with Heterogeneous Users and Selective Feedback: Individual Sum Rate and Individual Scaling Laws}
%
%
%

\author{Yichao Huang, \IEEEmembership{Member, IEEE}, and Bhaskar D. Rao, \IEEEmembership{Fellow,
IEEE}
\thanks{This research was supported by Ericsson endowed chair funds, the Center for Wireless Communications, UC Discovery grant com09R-156561 and NSF
grant CCF-1115645.}
\thanks{Y. Huang was with Department of Electrical and Computer Engineering, University
of California, San Diego, La Jolla, CA 92093-0407 USA. He is now with Qualcomm, Corporate R\&D, San Diego, CA, USA (e-mail: yih006@ucsd.edu).}
\thanks{B. D. Rao is with Department of Electrical and Computer Engineering,
University of California, San Diego, La Jolla, CA 92093-0407 USA (e-mail: brao@ece.ucsd.edu).}}

%
%

\markboth{IEEE TRANSACTIONS ON WIRELESS COMMUNICATIONS, to appear}%
{IEEE TRANSACTIONS ON WIRELESS COMMUNICATIONS, to appear}
%



\maketitle

\begin{abstract}
This paper investigates three open problems in random beamforming based communication systems: the scheduling policy with heterogeneous users,
the closed form sum rate, and the randomness of multiuser diversity with selective feedback. By employing the cumulative distribution function
based scheduling policy, we guarantee fairness among users as well as obtain multiuser diversity gain in the heterogeneous scenario. Under this
scheduling framework, the individual sum rate, namely the average rate for a given user multiplied by the number of users, is of interest and
analyzed under different feedback schemes. Firstly, under the full feedback scheme, we derive the closed form individual sum rate by employing a
decomposition of the probability density function of the selected user's signal-to-interference-plus-noise ratio. This technique is employed to
further obtain a closed form rate approximation with selective feedback in the spatial dimension. The analysis is also extended to random
beamforming in a wideband OFDMA system with additional selective feedback in the spectral dimension wherein only the best beams for the best-L
resource blocks are fed back. We utilize extreme value theory to examine the randomness of multiuser diversity incurred by selective feedback.
Finally, by leveraging the tail equivalence method, the multiplicative effect of selective feedback and random observations is observed to
establish the individual rate scaling.
\end{abstract}


%
\IEEEpeerreviewmaketitle

\section{Introduction}\label{introduction}
%
%
%
%

In multi-antenna downlink systems, transmission strategies which require less feedback resources \cite{love08, roh06, yoo06, yoo07} to fully
utilize multiuser diversity \cite{knopp95, viswanath02}, but with asymptotic sum capacity comparable to dirty paper coding \cite{caire03,
viswanath03, vishwanath03, yu04, weingarten06}, are favored. The idea of random beamforming \cite{sharif05}, which satisfies the two
aforementioned features has drawn much interest in recent years \cite{zhang07, tang07, al08, wang07, kim08, park09, moon11}. In the basic random
beamforming strategy suggested in \cite{sharif05}, the transmitter with $M$ transmit antennas generates $M$ random orthonormal beams and
requires each user to feed back the $\mathsf{SINR}$ experienced by them for each beam. Then the transmitter schedules users for transmission
that currently have the best channel for each random beam. Despite the considerable literature on this topic, there are three existing open
problems:
\begin{enumerate}
\item How to address heterogeneous users with diverse large scale channel effects and the impact on scheduling policy?
\item What is the closed form sum rate by exact\footnote[1]{We use the term exact to denote results valid for arbitrary but finite number of users as opposed to asymptotic results.} performance analysis?
\item What is the effect of selective feedback, both spatial and spectral, on the randomness of multiuser diversity?
\end{enumerate}

The first problem is related to a practical downlink system setting with asymmetrically located users having heterogeneous large scale channel
effects. This near-far effect was first treated in \cite{sharif05} by observing that the system becomes interference dominated when $M$ is large
enough. In the large $M$ setting, the authors prove that users are asymptotically equiprobable to be scheduled. However, when $M$ is finite and
not increasing simultaneously with the number of users, the greedy scheduling policy employed in \cite{sharif05} can not maintain fairness among
users. Also, if a round robin scheduling policy was utilized, fairness can be guaranteed, but no multiuser diversity gain could be achieved for
capacity growth. Therefore, an alternate scheduling policy is needed to maintain fairness while exploiting multiuser diversity at the same time.
In this paper, the cumulative distribution function (CDF)-based scheduling policy \cite{park05} is leveraged and analyzed in the random
beamforming framework, wherein the user whose rate for a given beam is high enough but least probable to become higher is selected. Under this
scheduling policy, each user can be equivalently viewed as competing with other users with the same CDF, thus making the study of individual
user rate more relevant and interesting than that of the sum rate. In this paper, we develop the notion of individual sum rate, which is the
individual user rate multiplied by the number of users, in order to demonstrate the multiuser diversity gain with user growth for a given user.

The second problem addresses exact system analysis, namely deriving closed form expression for the sum rate for arbitrary but finite number of
users. Note that even with full feedback, wherein each user conveys back the signal-to-interference-plus-noise ratio ($\mathsf{SINR}$) for $M$
beams, the closed form sum rate has not been derived. This is partially due to the complicated form of $\mathsf{SINR}$ and its interplay with
multiuser diversity. In this work, the problem is tackled and solved by a novel probability density function (PDF) decomposition \cite{huang11}
which decomposes and interprets the selected user's $\mathsf{SINR}$. In \cite{huang11}, the homogeneous setting is considered and in this paper,
the technique is extended to the heterogeneous user setting and the closed form individual sum rate is derived. The closed form result under
full feedback helps in evaluating the system performance and acts as the building block for exact analysis with selective feedback.

The third problem is concerned with standard selective feedback in the spatial dimension, wherein each user feeds back the $\mathsf{SINR}$ for
the best beam among the $M$ beams and the corresponding beam index. This selective feedback is fundamentally different than full feedback in two
aspects. The first difference is the two-stage maximization with the first stage maximization carried out by each user for feedback selection
and the second stage maximization carried out by the scheduler to perform user selection. Since the best beam is selected by each user, the
first stage maximization is over $M$ correlated $\mathsf{SINR}$. This correlation issue has been addressed in \cite{pugh10, yang11}, and the CDF
for the selected $\mathsf{SINR}$ at the user side is derived. In this paper, we propose an approximation for the CDF and utilize it to derive
closed form rate approximation. The other fundamental difference is the number of the $\mathsf{SINR}$ values that the scheduler has to maximize
over for each beam. This number is fixed and equals the number of users in the full feedback case. However, with selective feedback, it becomes
a random quantity. In other words, selective feedback results in a random effect on the multiuser diversity. This effect was first observed in
\cite{pugh10}. In this paper, we investigate the randomness of multiuser diversity by extremes over random samples and provide a rigorous
argument on the rate scaling.

The third problem is further  extended to include spectral selectivity  by examining a wideband OFDMA system, which groups the subcarriers into
resource blocks \cite{zhu09} to form the basic scheduling and feedback unit. In order to save feedback resource while not significantly
degrading the system performance, additional selective feedback in the spectral dimension is necessary. The effect of random beamforming in a
wideband system is examined in \cite{svedman04} by extensive simulations, and further studied from a utility function perspective with the
proportional-fair scheduler in \cite{svedman07}. In \cite{fakhereddin09}, analytical results on the asymptotic cluster size is provided. Apart
from the thresholding-based partial feedback strategy \cite{sanayei07}, the best-L selective feedback strategy \cite{hur11} is appealing and
utilized in practical systems such as LTE \cite{dahlman11}. In this paper, we employ the best-L selective feedback strategy to investigate
random beamforming and the effect of spectral dimension selective feedback, which calls for an additional maximization stage at the user side to
perform feedback selection. In this feedback strategy, only the best beams from the best $L$ resource blocks along with the beam and resource
block index are fed back from each user. In this paper, we first derive a closed form rate approximation with exact analysis, i.e., valid for
arbitrary but finite number of users. Then, the influence of the additional spectral dimension selective feedback on the type of convergence is
investigated with the technique of tail equivalence. Moreover, the multiplicative effect of selective feedback and random observations is
observed to establish the rate scaling.

To summarize, the main contributions of this paper are threefold: the utilization of CDF-based scheduling policy to address heterogeneous users
with random beamforming, the obtained closed form rate results with different selective feedback assumptions, and the asymptotic analysis on the
randomness of multiuser diversity incurred by selective feedback. These three contributions analytically examine the raised open problems, and
foster further understanding on random beamforming with heterogeneous users and selective feedback. The organization of this paper is as
follows. Section \ref{system} reviews the basic narrowband system model for random beamforming. The analysis for the full feedback case is
carried out in Section \ref{full}, and  for the spatial dimension selective feedback in Section \ref{spatial}. Section \ref{both} provides the
model for the wideband OFDMA with random beamforming, and examines the effect of additional spectral dimension selective feedback on rate
performance. Finally, Section \ref{conclusion} concludes the paper.

\section{System Model}\label{system}
We consider a multi-antenna narrowband Gaussian downlink channel with $K$ single antenna receivers and a transmitter equipped with $M$ antennas. A
block fading channel model with coherence interval $T$ is assumed. The random beamforming strategy employs $M$ random orthonormal vectors
$\boldsymbol\phi_m\in \mathbb{C}^{M\times 1}$ for $m=1,\ldots,M$, where the $\boldsymbol\phi_i$'s are drawn from an isotropic distribution
independently every $T$ channel uses \cite{sharif05}. Denoting $s_m(t)$ as the $m$th transmission symbol at time $t$,  the transmitted vector of
symbols at time $t$, represented by $\mathbf{s}(t)\in\mathbb{C}^{M\times 1}$, is given as
\begin{equation}\label{system_eq_1}
\mathbf{s}(t)=\sum_{m=1}^M \boldsymbol\phi_m(t)s_m(t), \quad t=1,\ldots,T.
\end{equation}
Let $y_k(t)$ be the received signal at the $k$th user, then
\begin{equation}\label{system_eq_2}
y_k(t)=\sum_{m=1}^M\sqrt{\rho_k}\mathbf{h}_k^{\dag}(t)\boldsymbol\phi_m(t)s_m(t)+v_k(t),
\end{equation}
where $\mathbf{h}_k\in\mathbb{C}^{M\times 1}$ is the complex channel vector which is assumed to be known at the receiver, $v_k$ is the additive
white noise, and the elements of $\mathbf{h}_k$ and $v_k$ are i.i.d. complex Gaussian with zero mean and unit variance $\mathcal{CN}(0,1)$. Note
that this channel assumption corresponds to the Rayleigh fading assumption for the small scale channel effect. From now on, the time variable
$t$ will be dropped for notational convenience. The total transmit power is chosen to be $1$, i.e., $\mathbb{E}[\mathbf{s}^{\dag}\mathbf{s}]=1$,
and thus the received signal-to-noise ratio ($\mathsf{SNR}$) of user $k$ is $\rho_k$. In a practical downlink setting, due to different
locations of users, the large scale channel effects $\rho_k$ which may consist of path loss and shadowing vary across users. From
(\ref{system_eq_2}), the $\mathsf{SINR}$ of the $k$th user for the $m$th transmit beam can be computed as
\begin{equation}\label{system_eq_3}
\mathsf{SINR}_{k,m}=\frac{|\mathbf{h}_k^{\dag}\boldsymbol\phi_m|^2}{M/\rho_k+\sum_{i\neq m}|\mathbf{h}_k^{\dag}\boldsymbol\phi_i|^2},\quad
m=1,\ldots,M.
\end{equation}
Denote $Z_{k,m}\triangleq\mathsf{SINR}_{k,m}$ for notational simplicity. Then for a given beam $m$, the $Z_{k,m}$'s are independent across users
$k$ but non-identically distributed due to different $\rho_k$. For a given user $k$, the $Z_{k,m}$'s are identically distributed and correlated.
Thus the beam index $m$ can be dropped in the expression for the PDF, which is computed in \cite{sharif05} as
\begin{equation}\label{system_eq_4}
f_{Z_k}(x)=\frac{e^{-\frac{M}{\rho_k}x}}{(1+x)^M}\left(\frac{M}{\rho_k}(1+x)+M-1\right)u(x),
\end{equation}
where $u(\cdot)$ is the Heaviside step function. The CDF of $Z_k$ is shown in \cite{sharif05} to be
\begin{equation}\label{system_eq_5}
F_{Z_k}(x)=\left(1-\frac{e^{-\frac{M}{\rho_k}x}}{(1+x)^{M-1}}\right)u(x).
\end{equation}

\section{Full Feedback Analysis}\label{full}
This section is devoted to the analysis for the full feedback case wherein each user feeds back the $\mathsf{SINR}$ for $M$ beams. Since under
full feedback, all the beams are fed back, the order statistics for each beam is over $K$ independent random variables. Thus this case is well
suited for illustration of the scheduling policy and the derivation of the individual sum rate.

\subsection{Scheduling Policy and Individual Sum Rate}\label{full_rate}
After receiving the $\mathsf{SINR}_{k,m}$ from user $k$ for beam $m$, the scheduler is ready to perform user selection. In a homogeneous
setting, selecting the user with the largest $\mathsf{SINR}$ for a given beam maintains fairness and obtains multiuser diversity gain. This
system was analyzed in our recent work \cite{huang11}. The work is now expanded to the more complex heterogeneous case. In a heterogeneous
setting, the greedy scheduling policy would be highly unfair for finite $M$. The round robin scheduling policy can maintain scheduling fairness,
but no multiuser diversity gain can be obtained. The proportional-fair scheduling policy \cite{viswanath02, svedman07} achieves the system
fairness in terms of system utility. However, under the scenario of inter-beam interference, the users' rates are coupled under the
proportional-fair scheduling policy. This coupled effect makes it very difficult, if not impossible, to develop further analytical
results\footnote[2]{Note that extensive simulation results have been provided regarding the use of proportional-fair scheduling policy under
random beamforming in existing literature such as \cite{svedman07}. However, the coupled effect of user's rate prevents further analysis and it
remains an open problem to theoretically understand the system performance of proportional-fair scheduling policy under the heterogeneous user
setting with inter-beam interference.}. Therefore, to tackle this problem it is useful to consider alternate scheduling policies that decouple
each user's rate. In this paper, we employ the CDF-based scheduling policy \cite{park05} for further analysis. According to this policy, the
scheduler will utilize the distribution of the received $\mathsf{SINR}$, i.e., $F_{Z_k}$. It is assumed that the scheduler perfectly knows the
CDF\footnote[3]{This is the only system requirement to conduct the CDF-based scheduling, and the CDF can be obtained by infrequent feedback from
users and learned by the system.}, and it performs the following transformation \cite{park05}:
\begin{equation}\label{full_eq_1}
\tilde{Z}_{k,m}=F_{Z_k}(Z_{k,m}).
\end{equation}
The transformed random variable $\tilde{Z}_{k,m}$ is uniformly distributed ranging from $0$ to $1$, and independent and identically distributed
(i.i.d.) across users for a given beam. Denote $k_m^*$ as the random variable representing the selected user for beam $m$, then
\begin{equation}\label{full_eq_2}
k_m^*=\mathop{\max}\limits_{\mathcal{U}_m}\;\tilde{Z}_{k,m},
\end{equation}
where $\mathcal{U}_m$ denotes the set of users conveying feedback for beam $m$. In the full feedback case, $|\mathcal{U}_m|=K$. After user
$k_m^*$ is selected per (\ref{full_eq_2}), the scheduler utilizes the corresponding $Z_{k_m^*,m}$ for rate matching of the selected user. Let
$X_m$ be the $\mathsf{SINR}$ of the selected user for beam $m$ and now consider the sum rate of the system defined as follows,
\begin{equation}\label{full_eq_3}
R=\mathbb{E}\left[\mathop{\sum}\limits_{m=1}^M\log_2\left(1+X_m\right)\right].
\end{equation}
From the aforementioned formulation, the sum rate can be computed in the following procedure
\begin{align}
R&\mathop{\simeq}\limits^{(a)}M\mathbb{E}_{k_m^*}\left[\int_0^1\log_2\left(1+F_{Z_{k_m^*,m}}^{-1}(x)\right)dx^K\right]\notag\\
\label{full_eq_4}&\mathop{=}\limits^{(b)}\frac{M}{K}\sum_{k=1}^K\int_0^{\infty}\log_2(1+t)d(F_{Z_{k}}(t))^K=\frac{M}{K}\sum_{k=1}^K\mathcal{J}_k(K),
\end{align}
where (a) follows from the sufficient small probability that multiple beams are assigned to the same user; (b) follows from the change of
variable $x=F_{Z_{k_m^*,m}}(t)$, the fair property of the CDF-based scheduling, and the following definition for $\mathcal{J}_k(\epsilon)$
with exponent
$\epsilon\in\mathbb{N}_+$:
\begin{equation}\label{full_eq_5}
\mathcal{J}_k(\epsilon)\triangleq\int_0^{\infty}\log_2(1+x)d(F_{Z_{k}}(x))^{\epsilon}.
\end{equation}

With the help of the CDF-based scheduling, each user feels as if the other users have the same CDF for scheduling competition \cite{park05}.
Therefore, each user's rate is independent of other users making it possible to consider or predict individual user's rate by only examining its
own CDF. It is clear that the scheduling policy is not only fair, but also acknowledges multiuser diversity at the same time. If we denote the
sum rate as the ``macro" level understanding of the system performance, then the individual user rate can be seen as the ``micro" level
understanding of the system performance since this performance metric examines the rate for any specific user and the sum rate can be directly
computed from the individual user rate from all the users. Thus, under the CDF-based scheduling policy, each user's rate can be examined
separately and this property serves as one building block for further analysis with selective feedback.

In order to demonstrate the multiuser diversity gain for each individual user, we define the individual sum rate $\hat{R}_k$ for user $k$ which
is the individual user rate $R_k$ multiplied by the number of users, as follows
\begin{equation}\label{full_eq_6}
\hat{R}_k\triangleq KR_k=M\mathcal{J}_k(K).
\end{equation}
The definition of the individual sum rate under the CDF-based scheduling policy makes it natural to examine the rate scaling for each user
separately, and also provide a ``micro" level understanding of the sum rate scaling. Compared with the sum rate and the individual user rate which can
be treated as performance metrics, the notion of individual sum rate can be regarded as the analytic metric for further scaling analysis.

Note that in the homogeneous setting, $\mathcal{J}_k(\epsilon)$ reduces to
$\mathcal{J}(\epsilon)\triangleq\int_0^{\infty}\log_2(1+x)d(F_Z(x))^{\epsilon}$. It is mentioned in previous works that the exact closed form
for $\mathcal{J}(\epsilon)$ is hard to obtain due to the coupled effect of $\mathsf{SINR}$ and multiuser scheduling. In the sequel, the closed
form expression for $\mathcal{J}_k(\epsilon)$ is obtained which is the key to computing the sum rate given by (\ref{full_eq_4}). The main
technique is employing the following proposed PDF decomposition which readily follows from \cite{huang11}.
\begin{lemma} \label{lemma_1}
(\textit{PDF Decomposition}) $d(F_{Z_k}(x))^{\epsilon}$ can be decomposed as
\begin{equation}\label{full_eq_7}
d(F_{Z_k}(x))^{\epsilon} = \epsilon\sum_{i=0}^{\epsilon-1}{\epsilon-1\choose
i}\frac{(-1)^{i}}{i+1}d\left(1-\frac{e^{-\frac{M(i+1)x}{\rho_k}}}{(1+x)^{(M-1)(i+1)}}\right).
\end{equation}
\end{lemma}

With the help of this PDF decomposition, $\mathcal{J}_k(\epsilon)$ can be computed in closed form using standard integration techniques whose
expression is presented in the following theorem.
\begin{theorem} \label{theorem_1}
(\textit{Closed Form of $\mathcal{J}_k$})
\begin{equation}\label{full_eq_8}
\mathcal{J}_k(\epsilon)=\frac{\epsilon}{\ln2}\sum_{i=0}^{\epsilon-1}{\epsilon-1\choose
i}\frac{(-1)^{i}}{i+1}\mathcal{I}\left(\frac{M(i+1)}{\rho_k}, (M-1)(i+1)+1\right),
\end{equation}
where $\mathcal{I}(\alpha,\beta)\triangleq \int_0^{\infty}\frac{e^{-\alpha x}}{(1+x)^{\beta}} dx$ whose closed form expression is presented in
Appendix \ref{appenA}.
\end{theorem}
\begin{proof}
The proof is given in Appendix \ref{appenA}.
\end{proof}
\textit{Remark:} A few remarks are in order. Firstly, the analytically useful PDF decomposition decouples the effect of multiuser diversity and
random beamforming, which facilitates the integration. The decomposition is general in that it can be applied to other channel models, though in
this paper the simple Rayleigh channel model is assumed to obtain the $\mathsf{SINR}$ statistics in (\ref{system_eq_5}). Secondly, the derived
closed from results for the individual sum rate and the sum rate only involve finite sums and factorials, which can readily be computed.
Moreover, the derived $\mathcal{J}_k(\epsilon)$ will be employed as a building block for rate computation in Section \ref{spatial} and Section
\ref{both} with selective feedback.

\subsection{Individual Scaling Laws}\label{full_scaling}
With homogeneous setting, the asymptotic sum rate scaling is of interest and has been established as $M\log_2\log_2 K$ \cite{sharif05} given the
$\mathsf{SINR}$ statistics in (\ref{system_eq_5}). It can be easily seen that the multiuser diversity gain is linear with respect to the number
of transmit antennas. With heterogeneous setting employing the CDF-based scheduling, the same technique can be applied to obtain the asymptotic
scaling for the individual sum rate $\hat{R}_k$ of user $k$. We now develop the notion of individual rate scaling and state the individual scaling
laws under full feedback through the following theorem.
\begin{theorem} \label{theorem_2}
(\textit{Individual Scaling Laws Under Full Feedback})
\begin{equation}\label{full_eq_9}
\mathop{\lim}\limits_{K\rightarrow\infty}\frac{\hat{R}_k}{M\log_2\log_2 K}=1.
\end{equation}
\end{theorem}
\textit{Remark:} It is seen from Theorem \ref{theorem_2} that users asymptotically follow the same scaling laws in the CDF-based scheduling
policy. The large scale channel effect $\rho_k$ is not written explicitly in (\ref{full_eq_9}) since it is a constant inside the $\log$ term. It
should briefly be noted that the rate scaling only measures the asymptotic trend when $K\rightarrow\infty$ and thus can not accurately match the
exact performance for finite regions of $K$.

\section{Selective Feedback in the Spatial Dimension}\label{spatial}
This section examines selective feedback in the spatial dimension wherein each user only conveys the best beam. This standard user side selection
requires the handling of correlated random variables and the random effect on observations, which are pursued in Section \ref{spatial_rate} and
Section \ref{spatial_scaling}.

\subsection{Individual Sum Rate}\label{spatial_rate}
With selective feedback, each user selects and feeds back the largest $\mathsf{SINR}$ among $M$ beams. As discussed in Section \ref{system}, the
$Z_{k,m}$'s are correlated random variables given $k$. Thus simple order statistics result can not be used to characterize the selected
$\mathsf{SINR}$ at user side. Denote $Y_{k,m^*(k)}=\mathop{\max}\limits_m Z_{k,m}$ representing the selected $\mathsf{SINR}$ for user $k$ with
$m^*(k)$ as the selected beam index. Then according to the derivation in \cite{pugh10, yang11}, the CDF of $Y_{k,m^*(k)}$ is shown to be
\begin{equation}\label{spatial_eq_1}
F_{Y_{k,m^*(k)}}(x)=\left(1-\sum_{\imath=1}^M\frac{[d_{\imath}(x)]_+^Me^{-\frac{2Mx}{\rho_kd_{\imath}(x)}}}{A_{\imath}(x)}\right)u(x),
\end{equation}
where $d_{\imath}(x)=\frac{2(1-(M-\imath)x)}{M-\imath+1}$, $A_{\imath}(x)=d_{\imath}(x)\prod_{i\neq \imath}^M(d_{\imath}(x)-d_i(x))$, and
$[\cdot]_+$ is the positive part of the argument. Note that the distribution does not depend on the selected beam index $m^*(k)$ due to the
identically distributed property across beams and is dropped to simplify notation, i.e., $ F_{Y_{k,m^*(k)}}(x) = F_{Y_{k}}(x)$. Using a similar
procedure to that described in Section \ref{full_rate}, after receiving feedback, the scheduler performs the transformation for user selection:
\begin{equation}\label{spatial_eq_2}
\tilde{Y}_{k,m^*(k)}=F_{Y_k}\left(Y_{k,m^*(k)}\right).
\end{equation}
Compared with (\ref{full_eq_1}), it is clear that $F_{Y_k}=F_{Z_k}$ for the full feedback case. Denote $k_m^*$ as the random variable
representing the selected user for beam $m$, then
\begin{equation}\label{spatial_eq_3}
k_m^*=\mathop{\max}\limits_{\mathcal{U}_m}\;\tilde{Y}_{k,m^*(k)},
\end{equation}
where $\mathcal{U}_m = \{ k : m^*(k) = m\}$ denotes the set of users conveying feedback for beam $m$. $\mathcal{U}_m$ is a set of random size
and the probability mass function (PMF) can be shown to be given by
\begin{equation}\label{spatial_eq_4}
\mathbb{P}(|\mathcal{U}_m|=\tau_1)={K\choose \tau_1}\left(\frac{1}{M}\right)^{\tau_1}\left(1-\frac{1}{M}\right)^{K-\tau_1},\quad 0\leq\tau_1\leq
K.
\end{equation}
Following the derivation in Section \ref{full_rate}, let $X_m$ be the selected $\mathsf{SINR}$ for beam $m$ at the scheduler side, then
conditioned on $k_m^*$ and $|\mathcal{U}_m|=\tau_1$, the conditional CDF of $X_m$ can be written as $F_{X_m\mid
k_m^*,|\mathcal{U}_m|=\tau_1}(x)=(F_{Y_{k_m^*,m}}(x))^{\tau_1}$. By averaging over the randomness of $|\mathcal{U}_m|$, the conditional CDF is
expressed as
\begin{equation}\label{spatial_eq_5}
F_{X_m\mid k_m^*}(x)=\sum_{\tau_1=0}^K{K\choose
\tau_1}\left(\frac{1}{M}\right)^{\tau_1}\left(1-\frac{1}{M}\right)^{K-\tau_1}(F_{Y_{k_m^*,m}}(x))^{\tau_1}.
\end{equation}
From (\ref{full_eq_4}) and (\ref{full_eq_6}), the individual sum rate of user $k$ is derived as\footnote[4]{In this paper, it is assumed that if
no user feeds back $\mathsf{SINR}$ for a certain beam, that beam would be in scheduling outage and would not contribute to rate calculation.}
\begin{equation}\label{spatial_eq_6}
\hat{R}_k=M\sum_{\tau_1=1}^K{K\choose
\tau_1}\left(\frac{1}{M}\right)^{\tau_1}\left(1-\frac{1}{M}\right)^{K-\tau_1}\int_0^{\infty}\log_2(1+x)d(F_{Y_{k}}(x))^{\tau_1}.
\end{equation}

\begin{figure}[t]
\centering
    \includegraphics[width=0.85\linewidth]{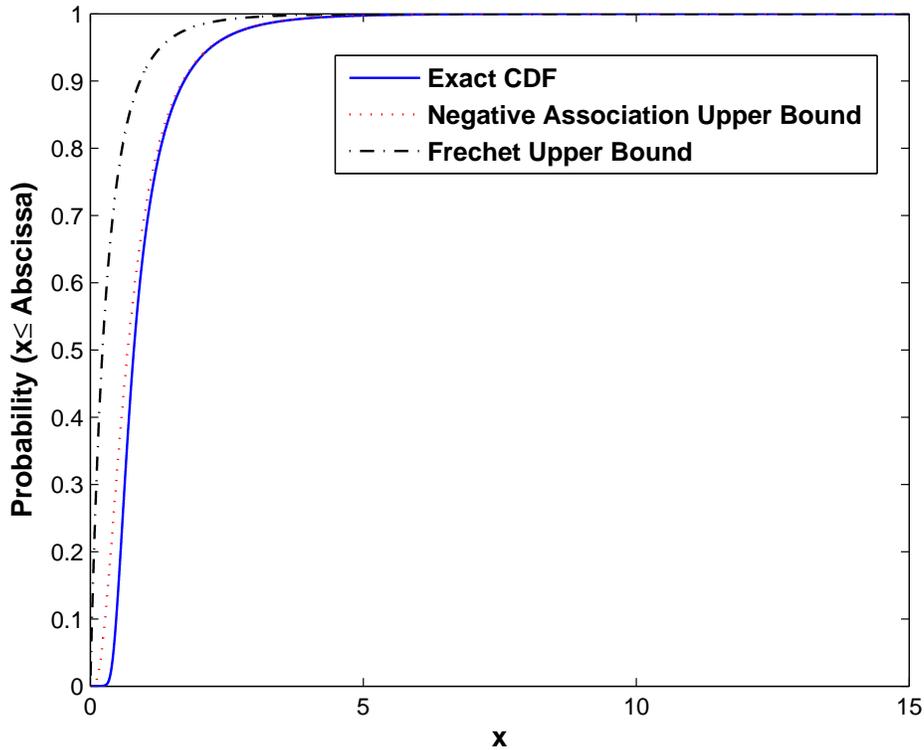}
\caption{Comparison of the exact CDF $F_{Y_k}$ with the Fr{\'e}chet upper bound and the negative association upper bound for spatial dimension
selective feedback ($M=4$, $\rho_k=10$ dB).} \label{fig_1}
\end{figure}

Due to the complicated form of $F_{Y_k}$, the exact closed form expression for (\ref{spatial_eq_6}) is hard to obtain. We now aim to provide an
approximate expression for the closed form  by examining the property of $F_{Y_k}$ and utilizing the established result in Section
\ref{full_rate}. Recall that $Y_k$ is the maximization over $M$ correlated random variables $Z_{k,m}$, thus alternative approximation for
$F_{Y_k}$ would lead to rate approximation. One simple approach is to use the Fr{\'e}chet upper bound \cite{galambos78} for the $Z_{k,m}$'s.
Since the $Z_{k,m}$'s are identically distributed across $m$, the Fr{\'e}chet upper bound yields $F_{Z_k}$. This upper bound is very loose
empirically for $F_{Y_k}$. One suitable approach is inspired by the conjectured negative associated upper bound proposed in \cite{pugh10} to
deal with the minimum mean square error (MMSE) receiver. Our empirical evidence shows that even with single antenna receiver, the $Z_{k,m}$'s
are negative associated \cite{joag83}, thus the upper bound produced by the negative association property can be utilized to approximate
$F_{Y_k}$, namely
\begin{equation}\label{spatial_eq_7}
F_{Y_k}(x)\simeq (F_{Z_k}(x))^M.
\end{equation}
Fig. \ref{fig_1} illustrates the bounds and the empirical CDF $F_{Y_k}$ for $M=4$, $\rho_k=10$ dB. It can be seen that the proposed upper bound
in (\ref{spatial_eq_7}) approximates the exact one in (\ref{spatial_eq_1}) well, especially when the $\mathsf{SINR}$ is large. By using the CDF
approximation, the individual sum rate can be approximated by a closed form expression presented in the following corollary.
\begin{corollary} \label{corollary_1}
(\textit{Closed Form Approximation of Individual Sum Rate})
\begin{equation}\label{spatial_eq_8}
\hat{R}_k\simeq \hat{R}_k^{\mathsf{App}}=M\sum_{\tau_1=1}^K{K\choose
\tau_1}\left(\frac{1}{M}\right)^{\tau_1}\left(1-\frac{1}{M}\right)^{K-\tau_1}\mathcal{J}_k(M\tau_1).
\end{equation}
\end{corollary}
\begin{proof}
The proof is given in Appendix \ref{appenB}.
\end{proof}
In order to demonstrate the rate approximation in Corollary \ref{corollary_1}, we conduct a numerical study in Fig. \ref{fig_2} for different
$M$ and $\rho_k$ with respect to the number of users. The exact $\hat{R}_k$ in (\ref{spatial_eq_6}) can be calculated by numerical integration.
It is observed that (\ref{spatial_eq_8}) approximates the exact rate very well, which makes the rate approximation valuable due to its efficient
computational form.

\begin{figure}[t]
\centering
    \includegraphics[width=1.0\linewidth]{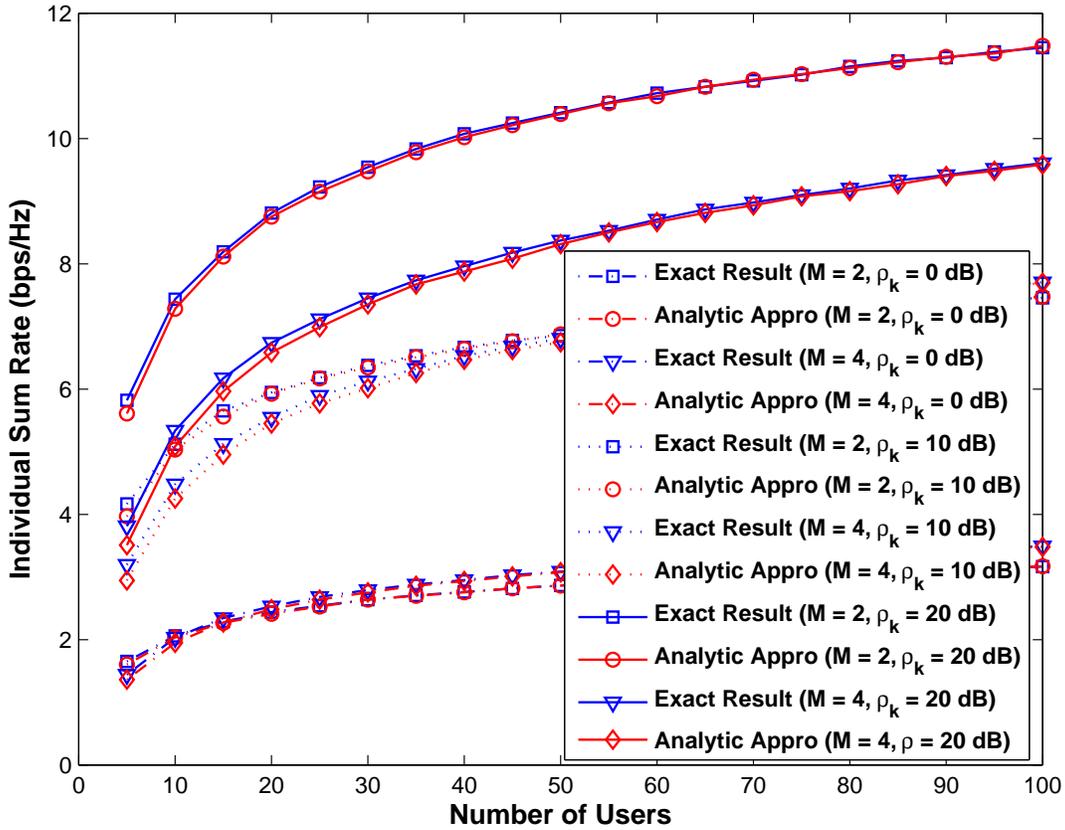}
\caption{Comparison of the exact individual sum rate and the approximated one for a given user with different $M$ and $\rho_k$ with respect to
the number of users ($M=2,4$, $\rho_k=0$ dB, $10$ dB, $20$ dB).} \label{fig_2}
\end{figure}

\subsection{Individual Scaling Laws}\label{spatial_scaling}
The difficulty of dealing with rate scaling with selective feedback is two-fold. Firstly, due to selective feedback of the best beam, the number
of $\mathsf{SINR}$ to maximize over at the scheduler side for each beam is a random quantity. This random effect is reflected in the random set
$\mathcal{U}_m$ in Section \ref{spatial_rate}. Secondly, the normalizing constants for establishing the type of convergence \cite{galambos78,
david03} have to be obtained for a quantity $\vartheta$ other than the number of users $K$ in the full feedback case. In \cite{pugh10}, the
first issue was tackled by the Delta method. In this paper, we solve the first issue by referring to the extremes over random samples, and
rigorously solve the second one by using the normalizing constants theorem. The proof is provided in Appendix \ref{appenB}.

To examine the random effect on multiuser diversity, denote the sequence of random variables $\kappa_m(K)$ as the number of $\mathsf{SINR}$ fed
back for beam $m$ with $K$ users. It is easy to see that $\kappa_m(K)$ are binomial distributed with probability of success $\frac{1}{M}$. Thus
by the strong law of large numbers, as $K$ grows, the number of $\mathsf{SINR}$ fed back for each beam becomes $\frac{K}{M}$. The following
theorem is called upon to deal with this random effect.
\begin{theorem} \label{theorem_3}
(\textit{Extremes with Random Sample Size \cite{galambos78, berman62})} Let, as $K\rightarrow\infty$, $\frac{\kappa(K)}{K}\rightarrow\vartheta$
in probability, where $\vartheta$ is a positive random variable. Assume that there are sequences $a_K\in \mathbb{R}, b_K>0$ such that
$\frac{\Lambda_K-a_K}{b_K}$ converges weakly to a nondegenerate distribution function $G$. Then, as $K\rightarrow\infty$,
\begin{equation}
\label{spatial_eq_9} \lim\mathbb{P}\left(\Lambda_{\kappa(K)}<a_K+b_K x\right)=\int_{-\infty}^{\infty}G^{y}(x)d\mathbb{P}(\vartheta<y).
\end{equation}
\end{theorem}
Therefore, if we denote $\Lambda_{k:\kappa(K)}$ as the extreme order statistics of the received $\mathsf{SINR}$ for each beam of a given user $k$,
then from Theorem \ref{theorem_3}, its CDF can be efficiently approximated by $(F_{Y_k})^{\frac{K}{M}}$. Combining this with the normalizing
constants theorem in Appendix \ref{appenB} yields the following corollary.
\begin{corollary} \label{corollary_2}
(\textit{Individual Scaling Laws Under Spatial Dimension Selective Feedback})
\begin{equation}\label{spatial_eq_10}
\mathop{\lim}\limits_{K\rightarrow\infty}\frac{\hat{R}_k}{M\log_2\log_2 \frac{K}{M}}=1,\quad\quad
\mathop{\lim}\limits_{K\rightarrow\infty}\frac{\hat{R}_k^{\mathsf{App}}}{M\log_2\log_2 K}=1.
\end{equation}
\end{corollary}
\begin{proof}
The proof is given in Appendix \ref{appenB}.
\end{proof}
\textit{Remark:} The scaling for the exact rate $\hat{R}_k$ and approximated rate $\hat{R}_k^{\mathsf{App}}$ differs in the factor
$\frac{1}{M}$. The rate scaling for $\hat{R}_k^{\mathsf{App}}$ does not have this factor because intuitively the exponent $M$ in the
approximated CDF $(F_{Z_k}(x))^M$ counteracts the reduction in the number of $\mathsf{SINR}$ values for maximization, i.e., $\frac{K}{M}$, due
to selective feedback. We call this effect as the \textit{multiplicative effect}. The detailed proof can be found in Appendix \ref{appenB}. To
draw further insights, we can think of the exponent of $F_{Z_k}(x)$ as the \textit{virtual} users. In the full feedback case, the exponent
equals $K$. In the selective feedback case with the approximated CDF, the exponent asymptotically equals $K$ by the aforementioned
multiplicative effect\footnote[5]{Note that even though the scaling laws are the same for the full feedback and the selective feedback case,
this metric only measures the asymptotic performance when $K$ is large. The exact rate performance is different due to the randomness of
multiuser diversity and the scheduling outage event.}. The notion of virtual users and the multiplicative effect will be investigated further
with both spatial and spectral dimension selective feedback in Section \ref{both_scaling}.

\begin{figure}[t]
\centering
    \includegraphics[width=0.9\linewidth]{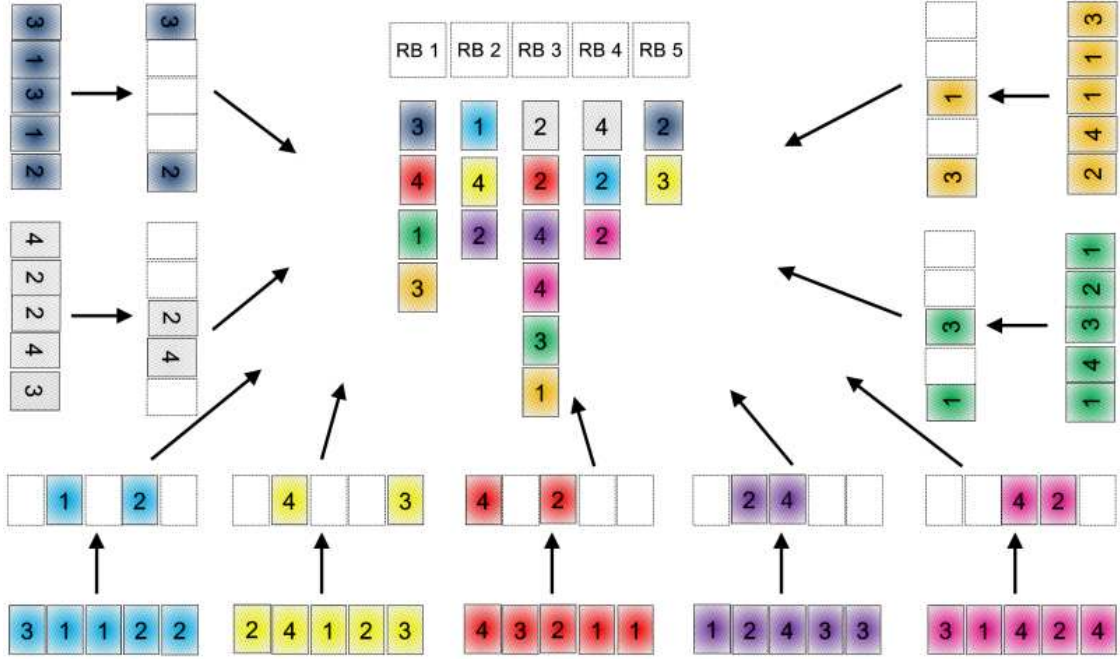}
\caption{Illustration of the spatial and spectral dimension selective feedback and the scheduling result in an OFDMA system (different colors
denote different users $K=9$, $N=5$ resource blocks, $M=4$ beams, the spectral dimension selective feedback $L=2$).} \label{fig_3}
\end{figure}

\section{Selective Feedback in Both Spatial and Spectral Dimension}\label{both}
In this section, random beamforming is embedded in a wideband OFDMA system. The system model is presented in Section \ref{both_model}, the exact
analysis and the asymptotic analysis are examined  in Section \ref{both_rate} and Section \ref{both_scaling} respectively.

\subsection{System Model}\label{both_model}
The system model described in Section \ref{system} is extended to an OFDMA system with $N$ resource blocks. Each resource block is regarded as
the basic scheduling and feedback unit. The random beamforming strategy generates $M$ orthonormal beams $\boldsymbol{\phi}_{m,n}$ for each
resource block. Denote $s_{m,n}$ as the $m$th transmission symbol at resource block $n$, then the received signal $y_{k,n}$ for user $k$ at
resource block $n$ can be expressed as
\begin{equation}\label{both_eq_1}
y_{k,n}=\sum_{m=1}^M\sqrt{\rho_k}\mathbf{H}_{k,n}^{\dag}\boldsymbol\phi_{m,n}s_{m,n}+v_{k,n},
\end{equation}
where $\mathbf{H}_{k,n}\in\mathbb{C}^{M\times 1}$ is the frequency domain channel transfer function of user $k$ at resource block $n$ with
i.i.d. $\mathcal{CN}(0,1)$ elements. To facilitate analysis, $\mathbf{H}_{k,n}$ is assumed to be i.i.d. across resource blocks for a given user.
This corresponds to the widely used block fading approximation in the frequency domain \cite{mceliece84, medard02} due to its simplicity and
capability to provide a good approximation to actual physical channels. The transmit power for a resource block is assumed to be $1$. From
(\ref{both_eq_1}), the $\mathsf{SINR}_{k,n,m}$ of user $k$ at resource block $n$ for beam $m$ is
$\mathsf{SINR}_{k,n,m}=\frac{|\mathbf{H}_{k,n}^{\dag}\boldsymbol\phi_{m,n}|^2}{M/\rho_k+\sum_{i\neq
m}|\mathbf{H}_{k,n}^{\dag}\boldsymbol\phi_{i,n}|^2}$, and is denoted by $Z_{k,n,m}$ for notational simplicity. For a given user $k$, the
$Z_{k,n,m}$'s are i.i.d. across resource blocks for a given beam $m$, and for a given resource block $n$, the $Z_{k,n,m}$'s are identically
distributed and correlated across beams. The CDF of $Z_{k,n,m}$ is given by
$F_{Z_k}(x)=\left(1-\frac{e^{-\frac{M}{\rho_k}x}}{(1+x)^{M-1}}\right)u(x)$, where the index $n$ and $m$ can be dropped due to the identically
distributed property.

\subsection{Individual Sum Rate}\label{both_rate}
With the extra degrees of freedom in the spectral dimension, additional selective feedback at each user side can be made possible by the
following two-stage feedback selection. The first stage selection is in the spatial dimension, where each user selects the best beam with the
largest $\mathsf{SINR}$ for each of the resource block. This process is similar to the narrowband feedback selection discussed in Section
\ref{spatial_rate}. Let $Y_{k,n,m}$ be the outcome of the first stage selection, thus from (\ref{spatial_eq_1}), its CDF can be written as
$F_{Y_k}(x)=\left(1-\sum_{\imath=1}^M\frac{[d_{\imath}(x)]_+^Me^{-\frac{2Mx}{\rho_kd_{\imath}(x)}}}{A_{\imath}(x)}\right)u(x)$, where again the
resource block index $n$ and the beam index $m$ can be dropped due to the identically distributed property across resource blocks and beams. The
second stage selection occurs in the spectral dimension, where each user feeds back the $\mathsf{SINR}$ values of the best $L$ resource blocks
among the total $N$ resource blocks. Let $W_{k,n,m}$ denote the outcome of the second stage selection of user $k$ at resource block $n$ for beam
$m$. Thus this random variable represents the selected $\mathsf{SINR}$ at the user side, whose CDF is of interest for further analysis. It is easy to see that for the case of full feedback in the
spectral dimension, i.e., $L=N$, $F_{W_k}=F_{Y_k}$. For the best-1 feedback case, i.e., $L=1$, $F_{W_k}=(F_{Y_k})^N$ due to
the independent property of $Y_k$ across resource blocks. For the general best-L feedback case, utilizing the results in \cite{hur11}, the CDF
can be shown as
\begin{equation}\label{both_eq_2}
F_{W_k}(x)=\sum_{\ell=0}^{L-1}\xi_1(N,L,\ell)(F_{Y_k}(x))^{N-\ell},
\end{equation}
where $\xi_1(N,L,\ell)=\sum_{i=\ell}^{L-1}\frac{L-i}{L}{N\choose i}{i\choose \ell}\left(-1\right)^{i-\ell}$. The two-stage feedback selection is
demonstrated in Fig. \ref{fig_3} with nine users denoted by different colors, five resource blocks, and four beams. In the illustrated example,
we use best-2 spectral dimension feedback, i.e., $L=2$.

After receiving feedback, the scheduler performs the CDF-based scheduling by first conducting the transformation on the received
$\mathsf{SINR}$,
\begin{equation}\label{both_eq_3}
\tilde{W}_{k,n,m}=F_{W_k}(W_{k,n,m}).
\end{equation}
Denote $k_{n,m}^*$ as the random variable representing the selected user at resource block $n$ for beam $m$, then
\begin{equation}\label{both_eq_4}
k_{n,m}^*=\mathop{\max}\limits_{\mathcal{U}_{n,m}}\;\tilde{W}_{k,n,m},
\end{equation}
where $\mathcal{U}_{n,m}$ denotes the set of users conveying feedback for beam $m$ at resource block $n$. Following the derivation in Section
\ref{spatial_rate}, let $X_{n,m}$ be the selected $\mathsf{SINR}$ for beam $m$ at resource block $n$ at the scheduler side. Then averaging over
the randomness of $|\mathcal{U}_{n,m}|$, the conditional CDF conditioned on $k_{n,m}^*$ can be written as
\begin{equation}\label{both_eq_5}
F_{X_{n,m}\mid k_{n,m}^*}(x)=\sum_{\tau_1=0}^K{K\choose
\tau_1}\left(\frac{1}{M}\right)^{\tau_1}\left(1-\frac{1}{M}\right)^{K-\tau_1}\sum_{\tau_2=0}^{\tau_1}{\tau_1\choose
\tau_2}\left(\frac{L}{N}\right)^{\tau_2}\left(1-\frac{L}{N}\right)^{\tau_1-\tau_2}(F_{W_{k_{n,m}^*,n,m}}(x))^{\tau_2}.
\end{equation}
For further derivation, $(F_{W_k}(x))^{\tau_2}$ is manipulated into the following form by the power series expansion \cite{gradshteyn07, hur11}:
\begin{equation} \label{both_eq_6}
(F_{W_k}(x))^{\tau_2}=\sum_{\ell=0}^{\tau_2(L-1)}\xi_2(N,L,\tau_2,\ell)(F_{Y_k}(x))^{N\tau_2-\ell},
\end{equation}
where
\begin{equation} \label{both_eq_7}
\xi_2(N,L,\tau_2,\ell)=\left\{
\begin{array}{l}
(\xi_1(N,L,0))^{\tau_2},\quad \ell=0\\ \frac{1}{\ell\xi_1(N,L,0)}\sum_{i=1}^{\min(\ell,L-1)}((\tau_2+1)i-\ell)\\
\quad\times\xi_1(N,L,i)\xi_2(N,L,\tau_2,\ell-i),\quad 1\leq \ell<\tau_2(L-1)\\ (\xi_1(N,L,L-1))^{\tau_2},\quad \ell=\tau_2(L-1).\\
\end{array} \right.
\end{equation}
Following the same procedure as in Section \ref{spatial_rate}, the individual sum rate for user $k$ can be derived as
\begin{align}
\hat{R}_k&=\frac{1}{N}\sum_{n=1}^N\mathbb{E}\left[\sum_{m=1}^M\log_2\left(1+X_{n,m}|k_{n,m}^*=k\right)\right]\notag\\
&=M\sum_{\tau_1=1}^K{K\choose
\tau_1}\left(\frac{1}{M}\right)^{\tau_1}\left(1-\frac{1}{M}\right)^{K-\tau_1}\sum_{\tau_2=1}^{\tau_1}{\tau_1\choose
\tau_2}\left(\frac{L}{N}\right)^{\tau_2}\left(1-\frac{L}{N}\right)^{\tau_1-\tau_2}\notag\\
\label{both_eq_8}&\quad\times \sum_{\ell=0}^{\tau_2(L-1)}\xi_2(N,L,\tau_2,\ell)\int_0^{\infty}\log_2(1+x)d(F_{Y_k}(x))^{N\tau_2-\ell}.
\end{align}
In order to obtain the closed form rate approximation for $\hat{R}_k$, the CDF approximation proposed in (\ref{spatial_eq_7}) by the negative
association property is utilized to approximate $F_{Y_k}$. The closed form result is presented in the following corollary.
\begin{corollary} \label{corollary_3}
(\textit{Closed Form Approximation of Individual Sum Rate})
\begin{align}
\hat{R}_k\simeq \hat{R}_k^{\mathsf{App}}&=M\sum_{\tau_1=1}^K{K\choose
\tau_1}\left(\frac{1}{M}\right)^{\tau_1}\left(1-\frac{1}{M}\right)^{K-\tau_1}\sum_{\tau_2=1}^{\tau_1}{\tau_1\choose
\tau_2}\left(\frac{L}{N}\right)^{\tau_2}\left(1-\frac{L}{N}\right)^{\tau_1-\tau_2}\notag\\
\label{both_eq_9}&\quad\times \sum_{\ell=0}^{\tau_2(L-1)}\xi_2(N,L,\tau_2,\ell)\mathcal{J}_k(M(N\tau_2-\ell)).
\end{align}
\end{corollary}
To understand the impact of spectral dimension selective feedback, we conduct a numerical study assuming $N=10$, $M=4$. Fig. \ref{fig_4} plots
the exact and approximated rate for different $L$ under $\rho_k=10$ dB with respect to the number of users. It can be seen that when the number
of users is small, there is a certain rate gap between selective feedback and full feedback. However, the gap becomes
negligible when the number of users increases. In Fig. \ref{fig_5}, the performance is observed for different $\rho_k$ for $K=20$. From the two
figures, we can see that the proposed rate approximation tracks the exact performance very well.

\begin{figure}[t]
\centering
    \includegraphics[width=0.9\linewidth]{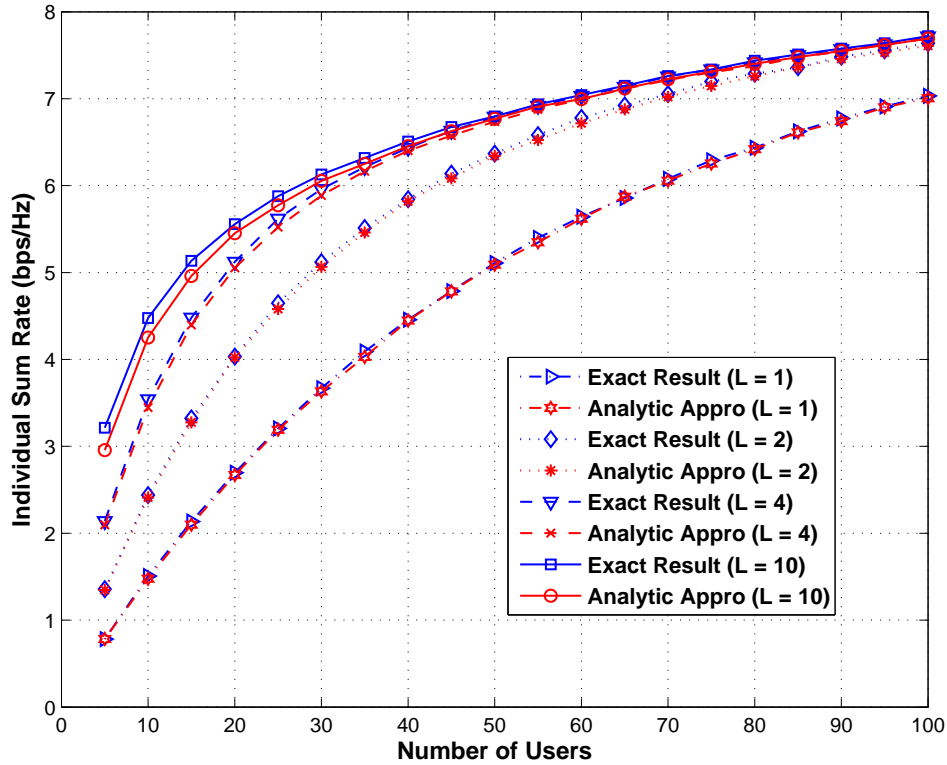}
\caption{Comparison of the exact individual sum rate and the approximated one for a given user with different spectral dimension selective
feedback $L$ with respect to the number of users ($M=4$, $N=10$, $\rho_k=10$ dB, $L=1,2,4,10$).} \label{fig_4}
\end{figure}

\begin{figure}[t]
\centering
    \includegraphics[width=0.9\linewidth]{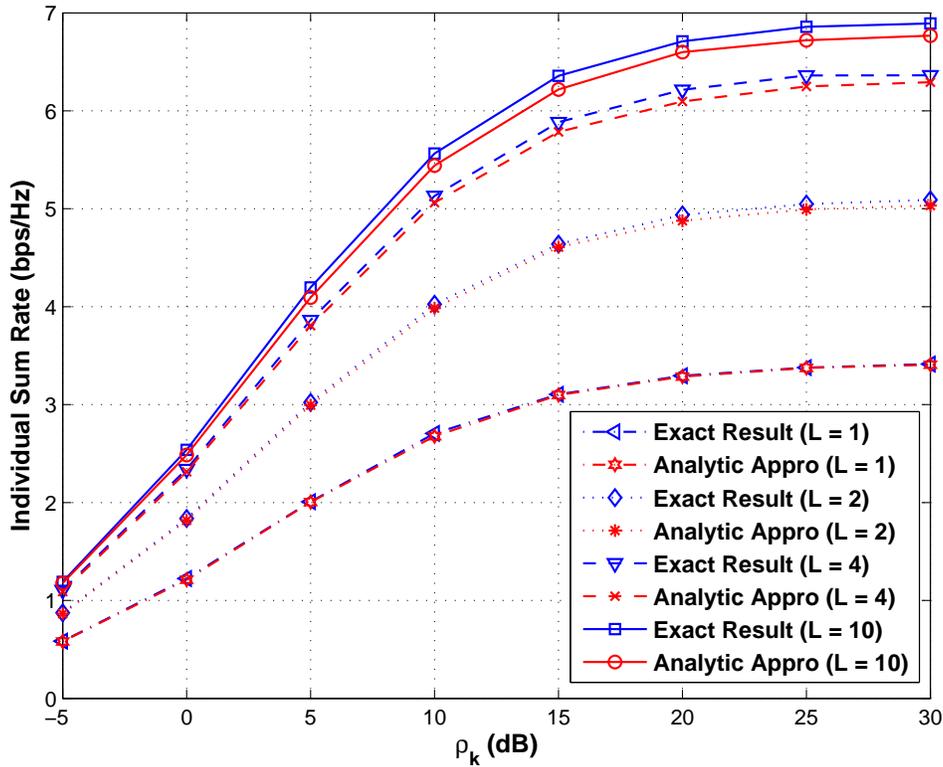}
\caption{Comparison of the exact individual sum rate and the approximated one for a given user with different spectral dimension selective
feedback $L$ with respect to different $\rho_k$ ($M=4$, $N=10$, $K=20$, $L=1,2,4,10$).} \label{fig_5}
\end{figure}

\subsection{Individual Scaling Laws}\label{both_scaling}
We now examine the rate scaling with selective feedback in both spatial and spectral dimension. In Section \ref{spatial_scaling} with spatial
dimension selective feedback, the CDF of interest is $F_{Y_k}$ and the number of $\mathsf{SINR}$ to maximize over at the scheduler side for each
beam approaches $\frac{K}{M}$. With additional spectral dimension feedback, the CDF of $F_{W_k}$ is of primary interest. To get a handle on the
randomness of multiuser diversity for this case, an  approach similar to that in Section \ref{spatial_scaling} can be utilized. Let the sequence of random
variables $\kappa_{n,m}(K)$ be the number of $\mathsf{SINR}$ values fed back for beam $m$ at resource block $n$ with $K$ users. It is easy to
see that $\kappa_{n,m}(K)$ are binomial distributed with probability of success $\frac{L}{MN}$. Therefore, by the strong law of large numbers,
as $K$ grows, the number of $\mathsf{SINR}$ values fed back for each beam at each resource block becomes $\frac{KL}{MN}$. Moreover, the
convergence property of the sequence $\kappa_{n,m}(K)$ can be shown by invoking the central limit theorem:
\begin{equation}\label{both_eq_10}
\mathop{\lim}\limits_{K\rightarrow\infty}\sqrt{K}\left(\frac{\kappa_{n,m}(K)}{K}-\frac{L}{MN}\right)\mathop{\rightarrow}^{d}\mathcal{N}\left(0,\frac{L}{MN}\left(1-\frac{L}{MN}\right)\right),
\end{equation}
where $d$ indicates convergence in distribution. By applying Theorem \ref{theorem_3}, the extreme order statistics of the received
$\mathsf{SINR}$ for each beam at each resource block for a given user $k$ can be efficiently approximated by
$\left(F_{W_k}\right)^{\frac{KL}{MN}}$.

\begin{figure}[t]
\centering
    \includegraphics[width=0.85\linewidth]{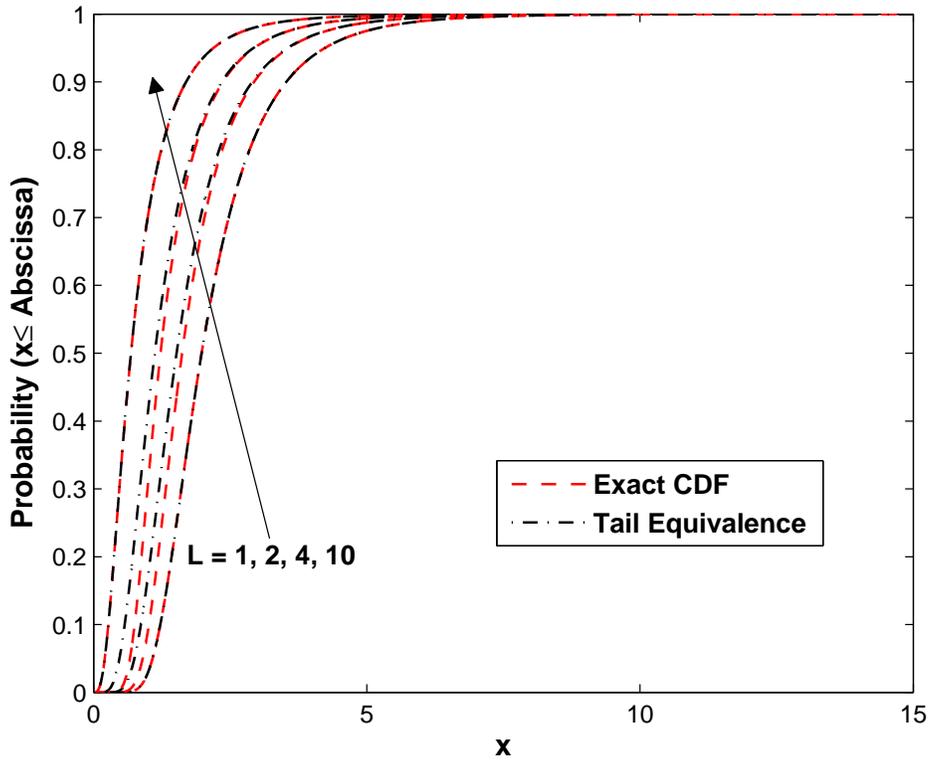}
\caption{Comparison of the exact CDF $F_{W_k}$ and its tail equivalence for different spectral dimension selective feedback $L$ ($M=4$, $N=10$,
$\rho_k=10$ dB, $L=1,2,4,10$).} \label{fig_6}
\end{figure}

Now the remaining problem is to examine the type of convergence of $F_{W_k}$. Recall the formulation of $F_{W_k}$ as:
$F_{W_k}(x)=\sum_{\ell=0}^{L-1}\xi_1(N,L,\ell)(F_{Y_k}(x))^{N-\ell}$. It is known that $F_{Y_k}$ converges weakly to the type $3$ Gumbel
distribution. Due to the complicated form of $\xi_1(\cdot,\cdot,\cdot)$, it is tedious to directly check the conditions for proving the type of
convergence. In order to investigate the tail behavior of $F_{W_k}$ which dominates the type of convergence \cite{david03}, the following tail
equivalence theorem is called upon.
\begin{theorem} \label{theorem_4}
(\textit{The Tail Equivalence Theorem \cite{resnick71})} $U(\cdot)$ and $V(\cdot)$ are distribution functions such that
\begin{equation}
\label{both_eq_11} \lim_{x\rightarrow\infty}\frac{1-U(x)}{1-V(x)}=1.
\end{equation}
If there exist normalizing constants $a_K$, $b_K>0$ such that $U^K(a_K+b_Kx)\rightarrow G(x)$, where $G(x)$ is non-degenerate, then
$V^K(a_K+b_Kx)\rightarrow G(x)$.
\end{theorem}
From Theorem \ref{theorem_4} one can infer that if two distribution functions are tail equivalent, then they belong to the domain of attraction of the same
type. Employing Theorem \ref{theorem_4}, a tail equivalent formulation can be obtained for $F_{W_k}$ expressed in the following corollary.
\begin{corollary} \label{corollary_4}
(\textit{Tail Equivalent CDF}) $F_{W_k}(x)$ is tail equivalent to $(F_{Y_k}(x))^{N-\sum_{\ell=0}^{L-1}\xi_1(N,L,\ell)\ell}$.
\end{corollary}
\begin{proof}
The proof is given in Appendix \ref{appenC}.
\end{proof}
Fig. \ref{fig_6} compares the exact CDF and the corresponding tail equivalence for different selective feedback $L$ under $M=4$, $N=10$, and
$\rho_k=10$ dB. The tail equivalent CDF is observed to track the exact one even when $x$ is small, which supports and lends confidence in the
power of the tail equivalence theorem. Therefore, the tail equivalence is used to study the type of convergence, which is expressed in the
following lemma.
\begin{lemma} \label{lemma_2}
(\textit{Type of Convergence of Selective Feedback}) Given the statistical property of $F_{Y_k}$ in (\ref{spatial_eq_1}), $F_{W_k}$ belongs to
the domain of attraction of type $3$ Gumbel distribution.
\end{lemma}
\begin{proof}
The proof is given in Appendix \ref{appenC}.
\end{proof}
Having obtained the type of convergence for $F_{W_k}$, the rate scaling result can be derived by referring to the normalizing constants theorem
in Appendix \ref{appenB}. The individual rate scaling is provided below.
\begin{theorem} \label{theorem_5}
(\textit{Individual Scaling Laws Under Spatial and Spectral Dimension Selective Feedback})
\begin{equation}\label{both_eq_12}
\mathop{\lim}\limits_{K\rightarrow\infty}\frac{\hat{R}_k}{M\log_2\log_2 \frac{(N-\sum_{\ell=0}^{L-1}\xi_1(N,L,\ell)\ell)L}{MN}K}=1,\quad\quad
\mathop{\lim}\limits_{K\rightarrow\infty}\frac{\hat{R}_k^{\mathsf{App}}}{M\log_2\log_2
\frac{(N-\sum_{\ell=0}^{L-1}\xi_1(N,L,\ell)\ell)L}{N}K}=1.
\end{equation}
\end{theorem}
\begin{proof}
The proof is given in Appendix \ref{appenC}.
\end{proof}
\textit{Remark:} For the exact rate $\hat{R}_k$, the ultimate equivalent CDF of interest is
$F_{Y_k}^{\frac{(N-\sum_{\ell=0}^{L-1}\xi_1(N,L,\ell)\ell)L}{MN}K}$, thus the exponent $\frac{(N-\sum_{\ell=0}^{L-1}\xi_1(N,L,\ell)\ell)L}{MN}K$
due to multiplicative effect can be seen as the virtual users for scheduling competition. This exponent is for the general best-L spectral
dimension feedback. For the full feedback $L=N$, since $\xi_1(N,N,\ell)$ equals $1$ for $\ell=N-1$ and $0$ otherwise, the CDF becomes
$F_{Y_k}^{\frac{K}{M}}$. For the best-1 feedback $L=1$, since $\xi_1(N,1,\ell)$ is $1$ for $\ell=0$ and $0$ otherwise, the CDF becomes
$F_{Y_k}^{\frac{K}{M}}$. Intuitively, the best-1 feedback is asymptotically optimal due to the same number of virtual users. In other words,
even though additional maximization reduces the average number of  variables for feedback, it counteracts this reduction by increasing the
exponent of the CDF. The number of virtual users is the limiting factor that dominates rate scaling. For the approximated rate
$\hat{R}_k^{\mathsf{App}}$, since the approximated CDF compensates for the spatial dimension selection by increasing the exponent, the rate
scaling differs by a factor of $M$.

\section{Conclusion}\label{conclusion}
In this paper, an analytical approach is used to investigate the problem of random beamforming with heterogeneous users and selective feedback.
The heterogenous user scenario corresponds to the practical scenario of potentially different large scale channel effects for different users.
We leverage the CDF-based scheduling policy to decouple each user's rate and thus theoretically examine the individual user rate. We develop the
notion of individual sum rate to analyze the rate scaling for each individual user. We focus our analysis in this work on theoretically
understanding the effect of selective feedback in both spatial and spectral dimensions. On the exact analysis part, extensive numerical results
show that our approximate expression for the  rate under selective feedback is effective and provides an efficient expression for computing the
exact rate. On the asymptotic analysis part, we develop the notion of virtual users and the multiplicative effect to explain the impact of
selective feedback on rate scaling. We further discover that the limiting factor for the rate scaling is the exponent for the ultimate CDF of
the selected $\mathsf{SINR}$ at the scheduler side. The extension of this work from single antenna users to multiple antenna users and more
generally the investigation of a multicell MIMO setup with advanced receiver design will be the subject of our future work.


%

\appendices
\section{}\label{appenA}
\textit{Proof of Theorem \ref{theorem_1}:} With the help of Lemma \ref{lemma_1}, $\mathcal{J}_k(\epsilon)$ can be computed as
\begin{align}
\mathcal{J}_k(\epsilon)&\mathop{=}\limits^{(a)}\frac{\epsilon}{\ln2}\sum_{i=0}^{\epsilon-1}{\epsilon-1\choose
i}\frac{(-1)^{i}}{i+1}\int_0^{\infty}\ln(1+x)d\left(1-\frac{e^{-\frac{M(i+1)x}{\rho_k}}}{(1+x)^{(M-1)(i+1)}}\right)\notag\\
\label{appen_eq_2} &\mathop{=}\limits^{(b)} \frac{\epsilon}{\ln2}\sum_{i=0}^{\epsilon-1}{\epsilon-1\choose
i}\frac{(-1)^{i}}{i+1}\int_0^{\infty}\frac{e^{-\frac{M(i+1)x}{\rho_k}}}{(1+x)^{(M-1)(i+1)+1}}dx,
\end{align}
where (a) follows from applying Lemma \ref{lemma_1}; (b) follows from integration by parts. The closed form result for
$\mathcal{I}(\alpha,\beta)$ in Theorem \ref{theorem_1} can be computed in a recursive manner \cite{gradshteyn07} and is presented as follows
\begin{equation}
\label{appen_eq_3} \mathcal{I}(\alpha,\beta)=\left\{
\begin{array}{l} \frac{(-1)^{\beta-1}\alpha^{\beta-1}e^{\alpha}E_1(\alpha)}{(\beta-1)!}+\mathop{\sum}\limits_{i=1}^{\beta-1}\frac{(i-1)!}{(\beta-1)!}(-1)^{\beta-i-1}\alpha^{\beta-i-1}, \quad \beta\geq2\\
e^{\alpha}E_1(\alpha), \quad \beta=1\\\end{array} \right.\\
\end{equation}
where $E_1(x)=\int_x^\infty \frac{e^{-t}}{t}dt$ is the exponential integral function of the first order \cite{abramowitz72}.

\section{}\label{appenB}
\textit{Proof of Corollary \ref{corollary_1}:}
\begin{align}
\hat{R}_k^{\mathsf{App}}&\mathop{=}\limits^{(a)} M\sum_{\tau_1=1}^K{K\choose
\tau_1}\left(\frac{1}{M}\right)^{\tau_1}\left(1-\frac{1}{M}\right)^{K-\tau_1}\int_0^{\infty}\log_2(1+x)d(F_{Z_k}(x))^{M\tau_1}\notag\\
\label{appen_eq_4} &\mathop{=}\limits^{(b)}M\sum_{\tau_1=1}^K{K\choose
\tau_1}\left(\frac{1}{M}\right)^{\tau_1}\left(1-\frac{1}{M}\right)^{K-\tau_1}\mathcal{J}_k(M\tau_1),
\end{align}
where (a) follows from the CDF approximation in (\ref{spatial_eq_7}); (b) follows from the definition and computation of
$\mathcal{J}_k(\epsilon)$.

\medskip

\textit{Proof of Corollary \ref{corollary_2}:} It is shown in \cite{pugh10} that $F_{Y_k}$ belongs to the domain of attraction of type $3$
Gumbel distribution \cite{david03}. Thus if the number of $\mathsf{SINR}$ to maximize over for each beam is fixed and equals the number of users
$K$, then the following equation holds: $\mathop{\lim}\limits_{K\rightarrow\infty}(F_{Y_k}(a_{k:K}+b_{k:K}x))^K=\Psi(x)$, where
$\Psi(x)=e^{-e^{-x}}$ is the type $3$ Gumbel distribution, $a_{k:K}$ and $b_{k:K}$ represent the normalizing constants for user $k$. From
Theorem \ref{theorem_3}, the number of $\mathsf{SINR}$ to maximize over for each beam approaches $\frac{K}{M}$. Let $c_{k:K}$ and $d_{k:K}$
denote the normalizing constants for user $k$ under the selective feedback case. Then the following equation holds:
$\mathop{\lim}\limits_{K\rightarrow\infty}(F_{Y_k}(c_{k:K}+d_{k:K}x))^{\frac{K}{M}}=\Psi(x)$. In order to obtain $c_{k:K}$ and $d_{k:K}$, the
following theorem is called upon.
\begin{theorem} \label{theorem_6}
(\textit{The Normalizing Constants Theorem \cite{galambos78})} Let $F_K(y)$ be a sequence of distribution functions. Let $a_K$, $b_K>0$, $c_K$,
and $d_K>0$ be sequences of real numbers such that, as $K\rightarrow\infty$,
\begin{equation}
\label{appen_eq_5} \lim F_K(a_K+b_Kx)=U(x),\quad \lim F_K(c_K+d_Kx)=V(x)
\end{equation}
for all continuity points $x$ of the limits, where $U(x)$ and $V(x)$ are nondegenerate distribution functions. Then, as $K\rightarrow\infty$,
the limits: $\lim \frac{d_K}{b_K}=B\neq 0$, $\lim \frac{c_K-a_K}{b_K}=A$ are finite, and $V(x)=U(A+Bx)$.
\end{theorem}
The spatial dimension selective feedback case possesses the following situation in Theorem \ref{theorem_6}: $F_K(x)=(F_{Y_k}(x))^K$,
$a_K=a_{k:K}$, $b_K=b_{k:K}$, $c_K=c_{k:K}$, $d_K=d_{k:K}$, $U(x)=\Psi(x)$, and $V(x)=(\Psi(x))^M$. The sequence of $a_{k:K}$ has been derived
in \cite{sharif05} as: $a_{k:K}=\rho_k\log_2 K-\rho_k(M-1)\log_2\log_2 K+o(1)$. A suitable choice of $b_{k:K}$ for type $3$ is $g_k(b_{k:K})$,
where $g_k(x)$ is the growth function for user $k$ defined by $g_k(x)\triangleq \frac{1-F_{Z_k}(x)}{f_{Z_k}(x)}$. Thus a suitable sequence is
$b_{k:K}=\rho_k$ for all $K$. Solving $(\Psi(x))^M=\Psi(A+Bx)$ yields $A=-\log M$, $B=1$. Therefore, by referring to Theorem \ref{theorem_6},
the normalizing constants can be derived to be: $c_{k:K}=\rho_k\log_2\frac{K}{M}-\rho_k(M-1)\log_2\log_2 K+o(1)$, and $d_{k:K}=\rho_k$ for all
$K$. Then by employing the Corollary A.1. in \cite{sharif05}, the individual rate for user $k$, namely $\hat{R}_k$ scales as
$M\log_2\log_2\frac{K}{M}$.

Regarding the approximated rate $\hat{R}_k^{\mathsf{App}}$, since the approximated CDF by negative association is $(F_{Z_k}(x))^M$ and the
number of $\mathsf{SINR}$ to maximize over approaches $\frac{K}{M}$, we have
$\mathop{\lim}\limits_{K\rightarrow\infty}(F_{Z_k}(c_{k:K}+d_{k:K}x))^{M\frac{K}{M}}=\mathop{\lim}\limits_{K\rightarrow\infty}(F_{Z_k}(c_{k:K}+d_{k:K}x))^K=\Psi(x)$.
Thus the normalizing constants $c_{k:K}=a_{k:K}$, and $d_{k:K}=b_{k:K}$, which enables the approximated rate $\hat{R}_k^{\mathsf{App}}$ to scale
as $M\log_2\log_2 K$.

\section{}\label{appenC}
\textit{Proof of Corollary \ref{corollary_4}:} Given Theorem \ref{theorem_4}, the following equality holds:
\begin{align}
&\lim_{x\rightarrow\infty}\frac{1-\sum_{\ell=0}^{L-1}\xi_1(N,L,\ell)(F_{Y_k}(x))^{N-\ell}}{1-(F_{Y_k}(x))^{N-\sum_{\ell=0}^{L-1}\xi_1(N,L,\ell)\ell}}\notag\\
\label{appen_eq_6}&\mathop{=}\limits^{(a)}\lim_{x\rightarrow\infty}\frac{\sum_{\ell=0}^{L-1}\xi_1(N,L,\ell)(N-\ell)(F_{Y_k}(x))^{N-\ell-1}f_{Y_k}(x)}{(N-\sum_{\ell=0}^{L-1}\xi_1(N,L,\ell)\ell)(F_{Y_k}(x))^{N-\sum_{\ell=0}^{L-1}\xi_1(N,L,\ell)\ell-1}f_{Y_k}(x)}\mathop{=}\limits^{(b)}1,
\end{align}
where (a) follows from the L'Hospital's rule; (b) follows from the fact that $\sum_{\ell=0}^{L-1}\xi_1(N,L,\ell)=1$.

\medskip

\textit{Proof of Lemma \ref{lemma_2}:} $F_{Y_k}$ with statistics in (\ref{spatial_eq_1}) belongs to the domain of attraction of type $3$.
It can be shown that for any distribution function $F(x)$  which  converges weakly to the limiting distribution, then its exponent form $F^{\epsilon}(x)$
has the same type of convergence \cite{galambos78}, $(F_{Y_k}(x))^{N-\sum_{\ell=0}^{L-1}\xi_1(N,L,\ell)\ell}$ belongs to the domain of
attraction of type $3$. Then by Theorem \ref{theorem_4}, $F_{W_k}$ belongs to the domain of attraction of type $3$.

\medskip

\textit{Proof of Theorem \ref{theorem_5}:} A procedure similar to that used in proving Corollary \ref{corollary_2} can be used here. Since the number of
$\mathsf{SINR}$ to maximize over for each beam at each resource block approaches $\frac{KL}{MN}$, and $F_{W_k}$ belongs to the domain of
attraction of type $3$, the following equation holds:
$\mathop{\lim}\limits_{K\rightarrow\infty}(F_{W_k}(c_{k:K}+d_{k:K}x))^{\frac{KL}{MN}}=\Psi(x)$. By referring to the tail equivalence theorem,
the equivalent equation is:
$\mathop{\lim}\limits_{K\rightarrow\infty}(F_{Y_k}(c_{k:K}+d_{k:K}x))^{\frac{KL(N-\sum_{\ell=0}^{L-1}\xi_1(N,L,\ell)\ell)}{MN}}=\Psi(x)$.
Applying Theorem \ref{theorem_6} yields the normalizing constants:
$c_{k:K}=\rho_k\log_2\frac{KL(N-\sum_{\ell=0}^{L-1}\xi_1(N,L,\ell)\ell)}{MN}-\rho_k(M-1)\log_2\log_2 K+o(1)$, and $d_{k:K}=\rho_k$ for all $K$.
Therefore, $\hat{R}_k$ scales as $M\log_2\log_2\frac{KL(N-\sum_{\ell=0}^{L-1}\xi_1(N,L,\ell)\ell)}{MN}$.

For the approximated rate $\hat{R}_k^{\mathsf{App}}$ using the approximated CDF $(F_{Z_k}(x))^M$ for $F_{Y_k}$, the following equation holds:
$\mathop{\lim}\limits_{K\rightarrow\infty}(F_{Z_k}(c_{k:K}+d_{k:K}x))^{\frac{KL(N-\sum_{\ell=0}^{L-1}\xi_1(N,L,\ell)\ell)}{N}}=\Psi(x)$. Using
the same line of arguments, $\hat{R}_k^{\mathsf{App}}$ scales as $M\log_2\log_2\frac{KL(N-\sum_{\ell=0}^{L-1}\xi_1(N,L,\ell)\ell)}{N}$.



\ifCLASSOPTIONcaptionsoff
  \newpage
\fi



%


\newpage

%

\begin{IEEEbiography}[{\includegraphics[width=1in,height=1.25in,clip,keepaspectratio]{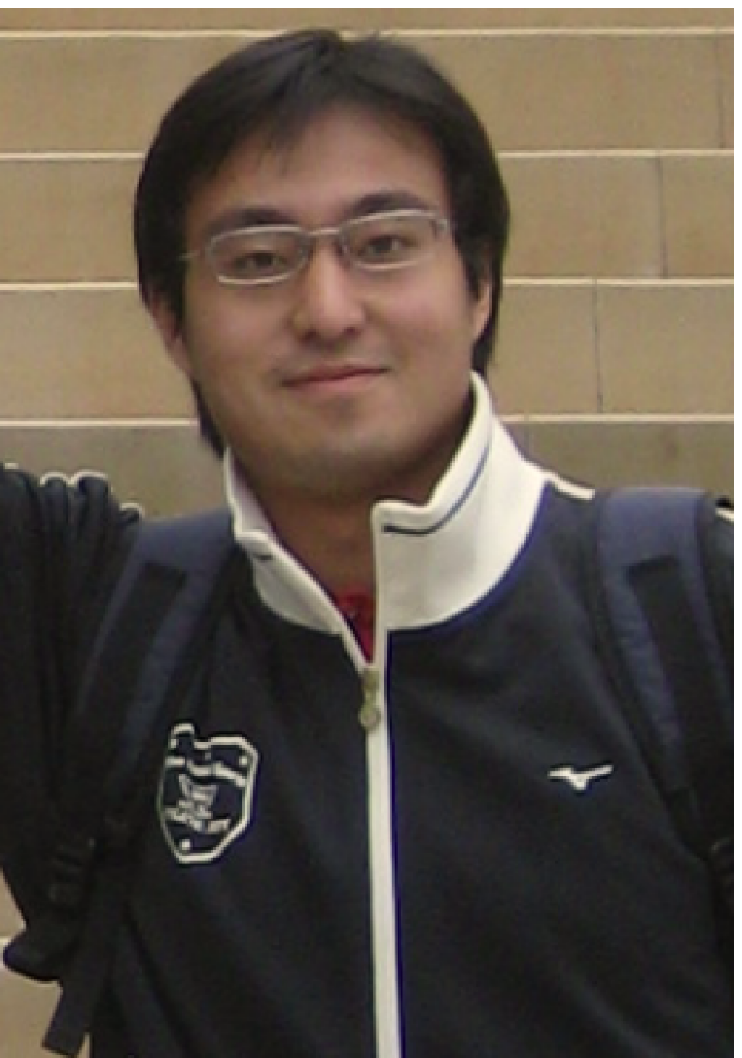}}]{Yichao Huang}
(S'10--M'12) received the B.Eng. degree in information engineering with highest honors from the Southeast University, Nanjing, China, in 2008,
and the M.S. and Ph.D. degrees in electrical engineering from the University of California, San Diego, La Jolla, in 2010 and 2012, respectively.
He then join Qualcomm, Corporate R\&D, San Diego, CA.

He interned with Qualcomm, Corporate R\&D, San Diego, CA, during summers 2011 and 2012. He was with California Institute for Telecommunications
and Information Technology (Calit2), San Diego, CA, during summer 2010. He was a visiting student at the Princeton University, Princeton, NJ,
during spring 2012. Mr. Huang received the Microsoft Young Fellow Award in 2007 from Microsoft Research Asia. He received the ECE Department
Fellowship from the University of California, San Diego in 2008, and was a finalist of Qualcomm Innovation Fellowship in 2010. His research
interests include communication theory, optimization theory, wireless networks, and signal processing for communication systems.
\end{IEEEbiography}
\begin{IEEEbiography}[{\includegraphics[width=1in,height=1.25in,clip,keepaspectratio]{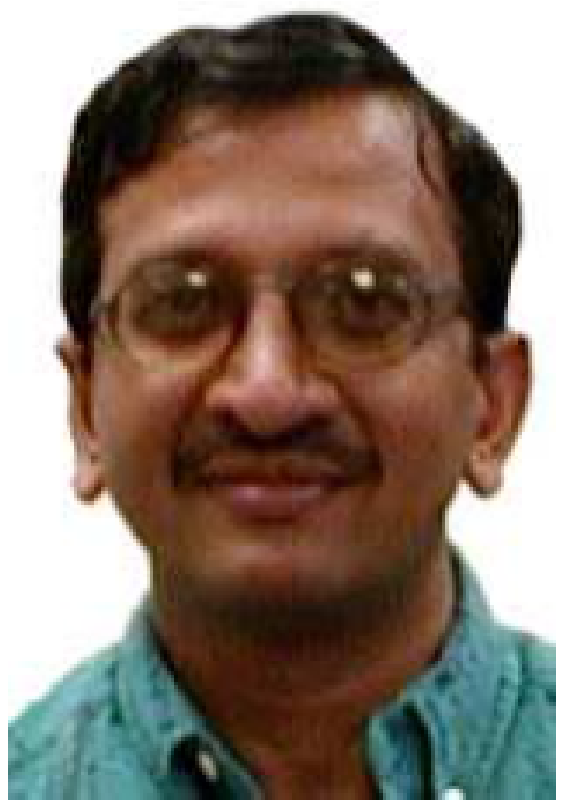}}]{Bhaskar D. Rao}
(S'80--M'83--SM'91--F'00) received the B.Tech. degree in electronics and electrical communication engineering from the Indian Institute of
Technology, Kharagpur, India, in 1979, and the M.Sc. and Ph.D. degrees from the University of Southern California, Los Angeles, in 1981 and
1983, respectively.

Since 1983, he has been with the University of California at San Diego, La Jolla, where he is currently a Professor with the Electrical and
Computer Engineering Department. He is the holder of the Ericsson endowed chair in Wireless Access Networks and was the Director of the Center
for Wireless Communications (2008--2011). His research interests include digital signal processing, estimation theory, and optimization theory,
with applications to digital communications, speech signal processing, and human--computer interactions.

Dr. Rao's research group has received several paper awards. His paper received the Best Paper Award at the 2000 Speech Coding Workshop and his
students have received student paper awards at both the 2005 and 2006 International Conference on Acoustics, Speech, and Signal Processing, as
well as the Best Student Paper Award at NIPS 2006. A paper he coauthored with B. Song and R. Cruz received the 2008 Stephen O. Rice Prize Paper
Award in the Field of Communications Systems. He was elected to the Fellow grade in 2000 for his contributions in high resolution spectral
estimation. He has been a Member of the Statistical Signal and Array Processing technical committee, the Signal Processing Theory and Methods
technical committee, and the Communications technical committee of the IEEE Signal Processing Society. He has also served on the editorial board
of the EURASIP Signal Processing Journal.
\end{IEEEbiography}




\end{document}